\documentclass[manuscript, screen]{acmart}
\AtBeginDocument{%
  }

\newcommand{\TODO}{\textcolor{red}{TODO}}
\newcommand{\bk}[1]{| #1 \rangle}
\usepackage{algorithm,algorithmic}

\setcopyright{acmlicensed}
\copyrightyear{2025}
\acmYear{2025}
\acmDOI{000000000}

\acmJournal{TQC}
\acmVolume{000000000}
\acmNumber{000000000}
\acmArticle{000000000}
\acmMonth{000000000}

\begin{document}

\title[Approximate $k$-Mismatch]{Simple Quantum Algorithm for Approximate $k$-Mismatch Problem}

\author{Ruhan Habib}
\email{ext.ruhan.habib@bracu.ac.bd}
\orcid{0009-0001-0911-9185}
\affiliation{%
  \institution{BRAC University}
  \city{Dhaka}
  \state{Dhaka}
  \country{Bangladesh}
}

\author{Shadman Shahriar}
\email{shadman.shahriar@bracu.ac.bd}
\orcid{0000-0001-8271-7272}
\affiliation{%
  \institution{BRAC University}
  \city{Dhaka}
  \state{Dhaka}
  \country{Bangladesh}
}

%
%
%
%
%
%

\renewcommand{\shortauthors}{Habib and Shahriar}

\begin{abstract}
    In the $k$-mismatch problem, given a pattern and a text of length $n$ and $m$ respectively,
    we have to find if the text has a sub-string with a Hamming distance of at most $k$ from the pattern.
    This has been studied in the classical setting since 1982 \cite{LANDAU1986239} and recently
    in the quantum computational setting by Jin and Nogler \cite{jin2022quantum}
    and Kociumaka, Nogler, and Wellnitz \cite{kociumaka2024quantum}.
    We provide a simple quantum algorithm that solves the problem in an approximate manner, given a parameter $\epsilon \in (0, 1]$.
    It returns an occurrence as a match only if it is a $\left(1+\epsilon\right)k$-mismatch.
    If it does not return any occurrence, then there is no $k$-mismatch. 
    This algorithm has a time (size) complexity of $\tilde{O}\left( \epsilon^{-1} \sqrt{\frac{mn}{k}} \right)$.
\end{abstract}



\received{\TODO}
\received[revised]{\TODO}
\received[accepted]{\TODO}

\maketitle

\section{Introduction}

String algorithms are of fundamental importance to Computer Science
both from a theoretical and practical point-of-views. They have
numerous applications in bio-informatics, data-mining and so on.
They are connected to both classical \cite{abboud2015quadratictime}
and quantum fine-grained complexity theory
\cite{legall_et_al.2022.97,akmal2021nearoptimal}.

The $k$-mismatch problem has been extensively studied in classical setting since 1982 \cite{LANDAU1986239},
but had not been studied through a quantum computational lens until 2022 by Jin and Nogler \cite{jin2022quantum}.
In that paper, they provided an $\tilde{O} \left( k\sqrt{n} \right)$-time
quantum algorithm and showed that the problem has a quantum query
lower-bound of $\Omega \left( \sqrt{kn} \right)$. They posed
the question of whether there is a quantum algorithm with better
query complexity than $\tilde{O} \left( k^{3/4} n^{1/2+o(1)} \right)$.
In 2024, Kociumaka, Nogler, and Wellnitz \cite{kociumaka2024quantum} found an algorithm
with optimal query complexity $\tilde{O}(\sqrt{kn})$ and time
complexity $\tilde{O} \left( \sqrt{n/m} (\sqrt{km} + k^2) \right)$.

In this paper, we show a simple quantum algorithm for an approximate variant of the $k$-mismatch problem:
given an approximation factor $\epsilon$, our algorithm has time complexity $\tilde{O}\left( \epsilon^{-1} \sqrt{\frac{mn}{k}} \right)$.
When $k = \omega\left({m^{2/3}} \epsilon^{-2/3}\right)$, our algorithm is faster than \cite{jin2022quantum}'s
quantum algorithm by a factor of $\omega\left(\sqrt{m}\right)$
and faster than that of \cite{kociumaka2024quantum} by a factor of $\omega(k)$.
A particular example is when $k$ is proportional to $m$.

\section{Notations and Basic Definitions}

We define $\mathbb{Z}$, $\mathbb{R}$, and $\mathbb{C}$ as usual:
set of integers, set of reals, and set of complex numbers.
We define $\mathbb{B} = \{ 0, 1\}$
and $\mathcal{B} = \mathbb{C}^2$.
Also, given a linear space $V$, we define $\mathbf{U}(V)$ to be the space of
unitary operators acting on $V$.

\begin{definition}[Intervals]
   Given two integers $L \le R$, we define
   $[L..R] = \{ x \in \mathbb{Z} : L \le x \le R  \}$,
      $[L..R) = [L..R] \setminus \{ R \}$,
      $(L..R] = [L..R] \setminus \{ L \}$,
      $(L..R) = [L..R] \setminus \{ L, R \}$.
   Given two real numbers $L \le R$, we define
      $[L, R] = \{x \in \mathbb{R} : L \le x \le R \}$, 
      $[L, R) = [L, R] \setminus \{R \}$,
      $(L, R] = [L, R] \setminus \{ L \}$,
      $(L, R) = [L, R] \setminus \{ L, R \}$.
\end{definition}

The $\tilde{O}$ (soft-oh) notation is used in place of the $O$ (big-oh) notation
to ignore polylogarithmic factors (for example, we can write $\tilde{O}\left(n^2\right)$
instead of $O\left(n^2 \log{n}\right)$).
Also note that we often use ``time complexity'' where we actually mean ``size complexity''.

%

For string and array indexing, we use $0$-based indexing. That is, the first letter
of a string $S$ is given by $S_0$ or $S[0]$. Furthermore, given integers $i$ and $j$,
$S[i..j]$ and $S[i..j)$ denotes substrings of $S$ starting from the $i$-th element (in $0$-based indexing)
to the $j$-th element or $j-1$th element respectively.
Also, $|S|$ denote the length of $S$.

\begin{definition}[Hamming Distance]
   Given two strings $A, B \in \Sigma^*$ for some
   alphabet $\Sigma$, we define $\delta_H(A,B)$ as follows:

   \begin{displaymath}
        \delta_H(A, B) = \begin{cases}
            \left| \{ i \in [ 0 .. |A| ) : A[i] \neq B[i] \} \right| & |A| = |B| \\
            \infty & |A| \neq |B|
        \end{cases}
    \end{displaymath}
\end{definition}

\section{Problem Statement and Our Contribution}

In the $k$-mismatch problem, the task is to find given a text and a pattern,
any substring of the text such that its Hamming distance with the pattern is
less than or equal to $k$. It is a ``fault-tolerant'' version of the regular
string matching problem.

\begin{definition}[$k$-mismatch Problem]
   An algorithm decides
   the \emph{$k$-mismatch matching problem} if, given
   oracle access
   a text string $T$ of length $n$, a pattern string $P$ of length
   $m$, and a positive integer $k$,
   the algorithm reports the existence of $i \in [0..n-m]$ such that
   $\delta_H(T[i..i+m), P) \le k$. 
   We also say that a substring $T'$ of $T$ is an $r$-mismatch of $P$ if
   $\delta_H\left(T', P\right) \le r$.
   A quantum algorithm decides the problem if, given $T$, $P$, and $k$ as defined above,
   it outputs a correct result (upon measurement) with probability of at least $2/3$.
\end{definition}

We solved an approximate version of this problem. Given a parameter $\epsilon \in (0, 1]$,
our algorithm is guaranteed (with probability of at least $2/3$) to return the
location of a $(1+\epsilon)k$-mismatch if there exists any $k$-mismatch. If there is
no $k$-mismatch, it may return the location of a $(1+\epsilon)k$-mismatch.
In any case, it will not (with probability of at least $2/3$) return the location
of any substring $T'$ with $\delta_H(T', P) > \left(1 + \epsilon\right)k$.

We also assume that the alphabet size of the strings are polynomially bounded:
each element of the text or pattern requires only polylogarithmically many
bits (or qubits) to be represented.

More formally, the following is the main result of our paper:
\begin{theorem}
   There exists a quantum algorithm that, given oracle access to 
   a pattern $P$ of length $m$ and
   a text $T$ of length $n$, an integer threshold $k > 0$, and
   $\epsilon \in (0, 1]$,
   such that:
   \begin{itemize}
      \item if there exists an $j \in [0..n-m]$ such that $\delta_H(T[j..j+m), P) \le k$, then 
         the algorithm, upon measurement, 
         outputs $(j', 1)$ for some $j' \in [0..n-m]$ satisfying 
         $\delta_H(T[j'..j'+m), P) \le (1 + \epsilon)k$ with a probability of at least $2/3$;
      \item if, for all $j \in [0..n-m]$, we have $\delta_H(T[j..j+m), P) > (1 + \epsilon)k$, then
         the algorithm, upon measurement, outputs $(j', 0)$ for some $j' \in [0..2n-1]$ 
         with probability of at least $2/3$.
   \end{itemize}
   This algorithm has time complexity $\tilde{O} \left( \epsilon^{-1} \sqrt{\frac{mn}{k}} \right)$ 
   (assuming that $P$ and $T$ can be accessed in $\tilde{O}\left(1\right)$ time).
\end{theorem}

The quantum algorithm outputs, upon measurement, a pair $(j,b)$. If $b = 1$, then
the algorithm reports $T[j..j+m)$ to be an $(1+\epsilon)k$-mismatch.
Otherwise, the algorithm reports that it did not find any $(1+\epsilon)k$-mismatch:
the value of $j$ does not matter in this case.

\section{Necessary Results}

The principle of deferred measurement is implicitly used throughout this paper.
Aside from that, the following results are also used.

\begin{theorem}[Amplitude Amplification \cite{Brassard_2002}]
    \label{thm:amplitude_amplification}
   There exists a quantum algorith $\mathbf{QSearch}$
   with the following property. Let $\mathcal{A}$ be any quantum algorithm
   (that uses no measurements), and let $\chi : \mathbb{B}^n \rightarrow \mathbb{B}$
   be a Boolean computable function. Also suppose that we are given oracle
   access or a quantum circuit for computing $\chi$. Let $a$
   denote the success probability of $\mathcal{A}$ (that is, the probability of $\mathcal{A}$,
   upon measurement, outputting $y$ such that $\chi(y) = 1$). Let $T$ be
   a positive parameter such that $a = 0$ or $T \ge 1/a$.
   If $a = 0$ then $\mathbf{QSearch}$ reports no answer.
   Otherwise, $\mathbf{QSearch}$ reports an answer in $O\left(\sqrt{T}\right)$ 
   applications of $\mathcal{A}$ and $\mathcal{A}^{-1}$
   with probability greater than or equal to $2/3$.
\end{theorem}

\begin{theorem}[Counting \cite{Brassard_2002}]
   \label{thm:qae_counting}
   Suppose that we are given positive integers $M$ and $k$, and a
   boolean (computable) function $f : [0..N-1] \rightarrow \mathbb{B}$,
   where $N = 2^n$ for some integer $n \ge 1$.
   There is a quantum algorithm $\mathbf{Count}(f, M)$ that
   outputs an estimate $t'$ to $t = \left| f^{-1}(1) \right|$
   such that 
   \begin{align*}
      \left| t' - t \right|
      \le
      2\pi k \frac{\sqrt{t(N - t)}}{M} + \pi^2 k^2 \frac{N}{M^2}
   \end{align*}
   with probability greater than $1 - \frac{1}{2(k-1)}$ for $k > 1$.
   Furthermore, this algorithm uses $f$ $\Theta(M)$ times.
\end{theorem}

\section{Our Result}

\subsection{Weak Search}

The Weak Search algorithm is heavily inspired by \cite{kociumaka2024quantum}'s 
Bounded-Error Quantum Seaarch with Neutral Inputs. In fact, the only difference is that
our assumption about the provided oracle is slightly more general.
It can also be though of as a simple application of Theorem \ref{thm:amplitude_amplification}.

To put it simply, suppose that we have access to some unitary circuit that, upon measurement,
outputs YES with probability of at least $2/3$ for some inputs (the positive inputs), outputs
NO with probability of at least $2/3$ for some inputs (the negative inputs),
and we do not necessarily know
how it behaves for the rest of the inputs (the neutral inputs).
The Weak Search algorithm finds, using $\tilde{O}\left(\sqrt{N} \right)$ queries,
either a positive input or a neutral input with probability of at least $2/3$
if any positive input exists. In any case, it reports a negative input with probability
of at most $1/3$.

\begin{theorem}[Weak Search]
    Let $n \ge 1$ be an integer and let $N = 2^n$. Let $F : [0..N-1] \rightarrow \{ 0, 1, 2 \}$
    be a function. Let $\mathcal{D}$ be a quantum circuit such that for any $j \in [0..N-1]$,
    if $F(j) = 0$ or $F(j) = 1$ then $| \langle j,  F(j), 0^x | \mathcal{D} | j, 0, 0^x \rangle | \ge 2/3$,
    where $x$ is the number of ancilliary qubits used by $\mathcal{D}$. Then
    there is a quantum circuit $\mathcal{B}$ such that $F^{-1}(\{ 1 \}) = \varnothing$
    or $\sum_{j \in F^{-1}(\{1,2\})} | \langle j, 1, 0^y | \mathcal{B} | 0, 0, 0^y\rangle|^2 \ge \frac{2}{3}$.
    In any case, $\sum_{j \in F^{-1}(\{0\})} |\langle j, 1, 0^y | \mathcal{B} | 0, 0, 0^y \rangle |^2 \le \frac{1}{3}$.
    Here, $y$ is the number of ancilliary qubits used by $\mathcal{B}$. Furthermore,
    $\mathcal{B}$ queries $\mathcal{D}$ at most $\tilde{O}\left(\sqrt{N} \right)$ times. And $\mathcal{B}$ increases
    the circuit size of $\mathcal{D}$ by a factor of $\tilde{O}\left(\sqrt{N}\right)$.
\end{theorem}

\begin{proof}
   Simply speaking, we are just using $\mathbf{QSearch}$
   on $\mathbf{Weak\_Search\_Auxiliary}$ (Algorithm \ref{alg:qtsearch0}), which samples $j \in [0..N-1]$
   and applies a boosted (decreasing the failure probability to $N^{-\lambda}$ for some $\lambda$ to be defined later)
   version of $\mathcal{D}$ on it. We call Algorithm \ref{alg:qtsearch0}'s output $(j,b)$, upon measurement,
   to be ``good'' or ``successful'' if $b = 1$.

   Suppose that $\mathcal{A} \in \mathbf{U}\left( \mathcal{B}^{\otimes m} \right)$
   and $\chi : \mathbb{B}^m \rightarrow \mathbb{B}$.
   Define $\mathbf{S}_0 \in \mathbf{U}\left( \mathcal{B}^{\otimes m} \right)$
   as follows: for all $x \in \mathbb{B}^n$,  if $x = 0$ then $\mathbf{S}_0\bk{x} = -\bk{x}$ and 
   $\mathbf{S}_0 \bk{x} = \bk{x}$ otherwise. Similarly, for all $x \in \mathbb{B}^n$,
   we define $\mathbf{S}_\chi \bk{x} = (-1)^{\chi(x)} \bk{x}$.
   Now, we define $\mathbf{Q}\left(\mathcal{A}, \chi \right) = -\mathcal{A}\mathbf{S}_0 \mathcal{A}^{-1}\mathbf{S}_\chi$.

   First, we write down the $\mathbf{QSearch}$' (Algorithm \ref{alg:qsearch1}), which is 
   just the Quantum Amplitude Amplification Algorithm of \cite{Brassard_2002}.

   Please see \cite{Brassard_2002}'s analysis of Theorem \ref{thm:amplitude_amplification},
   as our algorithm and analysis depends on theirs.

    \begin{algorithm}
        \caption{$\mathbf{QSearch'} (\mathcal{A}, \chi, T)$}\label{alg:qsearch1}
        \begin{algorithmic}[1]
            \STATE Set $l \leftarrow 0, t \leftarrow 0, f \leftarrow 0, o \leftarrow 0$ and let $c$ be any constant such that $1 < c < 2$.
            \STATE Set constant $L \leftarrow \max\left(C, \left\lceil \log{\left(4\alpha \sqrt{T} \right)}/\log{c}\right\rceil\right)$. \label{pseudo:qsearch1_L}
            \WHILE{$l < L$ and $f = 0$}
                \STATE Set $l \leftarrow l + 1$ and set $M \leftarrow \lceil c^l \rceil$.
                \STATE Set $t \rightarrow t + 1$.
                \STATE Apply $\mathcal{A}$ on the intial state of appropriate size $\bk{0}$.
                \STATE Measure the system, let $\bk{z,b}$ denote the outcome of the register on which $\mathcal{A}$ acts. \label{pseudo:qsearch1_z0}
                \IF{$\chi(z,b) = 1$}
                    \STATE Set $o \leftarrow (z, b)$ and $f \leftarrow 1$
                \ELSE
                \STATE Initialize a register of appropriate size to $\bk{\Psi} = \mathcal{A}\bk{0}$.
                    \STATE Pick an integer $j$ between $1$ and $M$ uniformly at random.
                    \STATE Set $t \leftarrow t + j$.
                    \STATE Apply $\mathbf{Q}(\mathcal{A}, \chi)^j$ to the register.
                    \STATE Measure the register, let $\bk{z,b}$ denote the outcome. \label{pseudo:qsearch1_z1}
                    \IF{$\chi(z,b) = 1$} 
                    \STATE Set $o \leftarrow (z, b)$ and $f \leftarrow 1$.
                    \ENDIF
                \ENDIF
            \ENDWHILE
            \STATE \RETURN $(o, f)$.
        \end{algorithmic}
    \end{algorithm}

    Let $a$ denote the success probability of $\mathcal{A}$.
    Let $T_f$ denote the random variable denoting the final value of $t$ in Algorithm
    \ref{alg:qsearch1} if we ignore the condition $l < L$ in the while loop.
    It can be shown (and has been shown in \cite{Brassard_2002}) that 
    if $a \ge 3/4$, $\mathbb{E}\left[T_f\right] \le C/3$ for some positive integer $C$.
    And if $0 < a < 3/4$, then  $\mathbb{E}\left[T_f\right] \ge \frac{\alpha}{4\sqrt{a}}$ 
    for some real $\alpha > 0$. This means that $\mathbb{P}\left[ T_f \le \alpha \sqrt{T} \right] \ge 3/4$.
    Furthermore, let $\gamma \in \mathbb{N}$ be a fixed constant
    such that $\mathbf{QSearch'}$ uses at most $\gamma \sqrt{N}$ applications
    of $\mathcal{A}$. And let $\lambda \ge 4$ be a fixed integer such that 
    $4\gamma 2^{-\lambda + \frac{1}{2}} \le \frac{1}{9}$.

    Let $\mathbf{Success\_Boosting}(\mathcal{A}, r, x)$ denote boosting
    the success of $\mathcal{A}$ to $1 - N^{-r}$ on input $x$, assuming of course
    that $\mathcal{A}$ is a ``decision'' quantum algorithm that outputs YES or NO
    correctly with probability of at least $2/3$. We can do this by simply
    computing $\mathcal{A}$ multiple times and taking a majority vote.

    \begin{algorithm}
        \caption{$\mathbf{Weak\_Search\_Auxiliary} (\mathcal{D}, N)$}\label{alg:qtsearch0}
        \begin{algorithmic}[1]
            \REQUIRE $N = 2^n$ for some integer $n \ge 1$.
            \STATE Sample $j$ uniformly randomly from $[0..N-1]$.
            \STATE Set $b \leftarrow \mathbf{Success\_Boosting}(\mathcal{D}, \lambda, j)$.
            \RETURN $(j, b)$.
        \end{algorithmic}
    \end{algorithm}


   Note that by replacing line 1 of of Algorithm \ref{alg:qtsearch0},
   with some other quantum algorithm , we can get a generalization of Theorem \ref{thm:amplitude_amplification}.

   For $n \in \mathbb{N}$, let $N = 2^n$ and define 
   $\chi_N : [0..N-1] \times \mathbb{B} \rightarrow \mathbb{B}$ by
   \begin{align*}
      \chi_N(j, b) = b &&\forall j \in [0..N-1], b \in \mathbb{B}.
   \end{align*}

   We are going to apply $\mathbf{QSearch'}$ on $\mathbf{Weak\_Search\_Auxiliary}$.
   An output $(j,b)$ of $\mathbf{Weak\_Search\_Auxiliary}$ is considered ``good''
   if $\chi_N(j,b) = b = 1$.

    \begin{algorithm}
        \caption{$\mathbf{Weak\_Search} (\mathcal{D}, N)$}\label{alg:qtsearch}
        \begin{algorithmic}[1]
            \REQUIRE $N = 2^n$ for some integer $n \ge 1$.
            \FOR{$t \in [0..1]$}
            \STATE Set $((j, b), f) \leftarrow \mathbf{QSearch'}\left(\mathbf{Weak\_Search\_Auxiliary}\left(\mathcal{D}, N \right),\chi_N, 2N \right)$.
                \IF{$f = 1$}
                    \RETURN $(j, f)$.
                \ENDIF
            \ENDFOR
            \STATE \RETURN $(0, 0)$.
        \end{algorithmic}
    \end{algorithm}

   Note that in Algorithm \ref{alg:qtsearch} we are passing the quantum circuit that
   computes $\mathbf{Weak\_Search\_Auxiliary}\left(\mathcal{D}, N\right)$ as an oracle to $\mathbf{QSearch}'$.
   
   Let $\mathcal{D}$, $N$, $F$ be given.

   Let $L$ be the constant defined in line \ref{pseudo:qsearch1_L} of Algorithm \ref{alg:qsearch1}.
   Let $Z_0, \dots, Z_{4L-1}$ and $B_0, \dots, B_{4L-1}$ be random variables 
   for each measured $\bk{z, b}$
   (line \ref{pseudo:qsearch1_z0} and \ref{pseudo:qsearch1_z1}).
   Note that there are $2 \cdot 2L$ indices for the random variables, because
   we are computing $\mathbf{QSearch'}$ twice.

   Let $F'$ and $J'$ be the random variable for the final output
   of $\mathbf{Weak\_Search}$. Then,
   $$
      \mathbb{P}[F' = 1 \land F(J') = 0]
      \le \sum_{j=0}^{4L-1} \mathbb{P}[F(Z_j) = 0 \land B_j = 1] 
      \le \sum_{j=0}^{4L-1} N^{-\lambda}
      = (4L-1)N^{-\lambda} 
    $$

   Using the fact that $L \le \gamma \sqrt{N}$
   for large enough $N \ge 2$,
   we have 
   $$
      \mathbb{P}[F' = 1 \land F(J') = 0]
      \le 4L N^{-\lambda} 
      \le \gamma 4 N^{-\lambda+\frac{1}{2}} 
      \le 4 \gamma 2^{-\lambda + \frac{1}{2}} 
      \le \frac{1}{9}
    $$

   In other words, we have shown that
   $
      \sum_{j \in F^{-1}(\{0\})} \left|\langle j, 1, 0^y | \mathcal{B} | 0, 0, 0^y \rangle \right|^2
      \le \frac{1}{9}
      \le
      \frac{1}{3}.
      $
   
   Now, suppose that $F^{-1}(\{1\}) \ne \varnothing$.
   Let $J$ and $B$ be random variables denoting the $j$ and $b$ from Algorithm \ref{alg:qtsearch0}.
   Then
   $$
      \mathbb{P}[B = 1]
      \ge
      \mathbb{P}[F(J) = 1 \land B = 1] 
      = \mathbb{P}[F(J) = 1] \cdot \mathbb{P}[B = 1 | F(J) = 1] 
      \ge \frac{1}{N} \cdot \left( 1 - N^{-\lambda} \right)
    $$
   For $N \ge 2$, $N^{-\lambda} \le N^{-4} \le \frac{1}{16}$ 
   and thus 
   $
      \mathbb{P}[B = 1] \ge \frac{1}{N} \left(  1 - N^{-\lambda} \right)
      \ge
      \frac{15}{16N}
      \ge
      \frac{1}{2N}
    $

   So the $a$ (the success probability) for $\mathbf{Weak\_Search\_Auxiliary}$ is bounded below
   by $\frac{1}{2N}$ (when $N \ge 2$). Since we are repeating $\mathbf{QSearch}'$ twice,
   we have $\mathbb{P}[F' = 1] \ge 1 - \frac{1}{3} \cdot \frac{1}{3} = \frac{8}{9}$ due to
   Theorem \ref{thm:amplitude_amplification}.

   Using the fact that $\mathbb{P}[F' = 1 \land F(J') = 0] \le \frac{1}{9}$,
   we get
   \begin{align*}
      \frac{8}{9} 
      &\le \mathbb{P}[F' = 1] 
      = \mathbb{P}\left[F' = 1 \land F(J') = 0\right] + \mathbb{P}[F' = 1 \land F(J') \in \{1, 2\}] 
      \le \frac{1}{9} + \mathbb{P}\left[F' = 1 \land F(J') \in \{1, 2\}\right] \\
      \frac{7}{9} 
      &\le \mathbb{P}[F' = 1 \land F(J') \in \{1, 2\}]
   \end{align*}

   In other words, if $F^{-1}\left( \{ 1 \} \right) \ne \varnothing$, then
   $
      \sum_{j \in F^{-1}(\{ 1, 2 \})} \left| \langle j, 1, 0^y | \mathcal{B} | 0, 0, 0^y \rangle \right|^2
      \ge \frac{7}{9}
      \ge \frac{2}{3} .
    $
\end{proof}

\subsection{Approximate Bounded Hamming Distance Pattern Matching}

The following is a generalization of Lemma 3.12 from \cite{kociumaka2024quantum} and its proof.
\begin{theorem}
   \label{thm:epsilon_k_distance}
   There is a quantum algorithm that,
   given oracle access to two strings $X$ and $Y$ of equal length $|X| = |Y| = m$,
   an integer threshold $k > 0$, and $\epsilon \in (0, 1]$,
   outputs YES ($1$) or NO ($0$) so that
   \begin{itemize}
      \item If $\delta_H(X, Y) \le k$, then the algorithm outputs YES with probability
         of at least $9/10$.
      \item If $\delta_H(X, Y) > (1 + \epsilon)k$, then the algorithm outputs NO with probability
         of at least $9/10$.
   \end{itemize}
   This algorithm takes $\tilde{O}\left(\epsilon^{-1} \sqrt{m / k}\right)$ quantum time.
\end{theorem}

\begin{proof}
   First, we present the quantum algorithm (Algorithm \ref{alg:approx_bounded_decider}).
   \begin{algorithm}
      \caption{$\mathbf{ApproxBoundedHammingDecider} (X, Y, k, \epsilon)$}\label{alg:approx_bounded_decider}
      \begin{algorithmic}[1]
         \STATE Set $m \leftarrow |X|$.
         \STATE Set $N \leftarrow \min\{ 2^j : j \in \mathbb{N} \land 2^j \ge m \}$.
         \STATE \textbf{procedure} \textsc{F}(j)
         \STATE \hspace{\algorithmicindent} \textbf{return} $j < m \land X_j \ne Y_j$.
         \STATE \textbf{end procedure}
         \STATE Set $M \leftarrow \left\lceil \frac{6\pi\sqrt{N/k}}{\sqrt{1+3\epsilon/2} - \sqrt{\epsilon}} \right\rceil$.
         \IF{$k \ge m$}
            \RETURN $1$
         \ELSE
         \STATE Set $t' \leftarrow \mathbf{Count}(\textsc{F}, M)$.
            \RETURN $t' < \left( 1 + \frac{\epsilon}{2} \right) k$.
         \ENDIF
      \end{algorithmic}
   \end{algorithm}

   If $k \ge m$, then the algorithm correctly returns YES (or $1$, to be precise).
   Otherwise, the algorithm outputs YES if and only if
   $t' \le \left(1 + \frac{\epsilon}{2}\right)k$.

   For the rest of the proof, assume that $k < m$.

   Instead of using Theorem \ref{thm:qae_counting}
   with parameters $\left(\left\lceil 48\pi \sqrt{N/k} \right\rceil, 6\right)$ as done in \cite{kociumaka2024quantum},
   we use parameters 
   $$ \left( \left\lceil \frac{6\pi \sqrt{N / k}}{\sqrt{1 + 3\epsilon/2} - \sqrt{1 + \epsilon}} \right\rceil, 6 \right) $$
   and with $F$ as the Boolean predicate. 

   Let $\beta = \sqrt{1 + 3\epsilon/2} - \sqrt{1 + \epsilon}$
   and $\alpha = 6\pi/\beta$. Then our first parameter is $M = \left\lceil \alpha \sqrt{N / k} \right\rceil$.
   Calculating, we get
   $
      \beta^2 + 2 \beta \sqrt{1 + \epsilon} = \frac{\epsilon}{2}
    $.
   Let $t$ denote the actual number of mismatches and let $t'$
   be a possible output by the counting algorithm. By Theorem 
   \ref{thm:qae_counting},
   we have
   \begin{align*}
      |t' - t| \le 12\pi \frac{\sqrt{t(N - t)}}{M} + \frac{36\pi^2 N}{M^2}
      \le 12\pi \frac{\sqrt{tN}}{M} + \frac{36\pi^2 N}{M^2}
   \end{align*}

   We shall show that if $t \le k$ then $t' \le (1 + \epsilon/2)k$
   and if $t > (1 + \epsilon)k$ then $t' > (1 + \epsilon/2)k$.

   First, suppose that $t \le k$. Then,
   \begin{align*}
      t' &\le t + 12\pi \frac{\sqrt{tN}}{\left\lceil \alpha \sqrt{\frac{N}{k}} \right\rceil} + 36\pi^2 \frac{N}{\left\lceil \frac{\alpha N}{k} \right\rceil^2}
         \le t + 12\pi \frac{\sqrt{tN}}{\alpha \sqrt{\frac{N}{k}}} + 36\pi^2 \frac{N}{\left( \frac{\alpha^2 N}{k} \right)} 
         = k + 12\pi \frac{\sqrt{tN}}{\alpha \sqrt{\frac{N}{k}}} + (6\pi/\alpha)^2 k \\
         &= k + 2\beta \sqrt{kt} + \beta^2 k
         \le \left( 1 + 2\beta + \beta^2 \right) k 
         = \left( 1 + 2\beta \sqrt{1 + \epsilon} + \beta^2 \right) k - 2\beta \left( \sqrt{1 + \epsilon} - 1 \right) k \\
         &< \left(1 + 2\beta\sqrt{1 + \epsilon} + \beta^2 \right) k 
         = (1 + \epsilon/2)k
   \end{align*}

   Now, suppose that $t > (1 + \epsilon)k$. Then,
   \begin{align*}
      t' 
      &\ge t - \left( 12\pi \frac{\sqrt{tN}}{M} + \frac{36\pi^2 N}{M^2} \right) 
      \ge t - \left( 12\pi \frac{\sqrt{tN}}{\left\lceil \alpha \sqrt{\frac{N}{k}} \right\rceil} + 36\pi^2 \frac{N}{\left\lceil \frac{\alpha N}{k} \right\rceil^2} \right) 
      \ge t - \left(12\pi \frac{\sqrt{tN}}{\alpha \sqrt{\frac{N}{k}}} + 36\pi^2 \frac{N}{\left( \frac{\alpha^2 N}{k} \right)} \right) \\
      &\ge t - \left(2\beta \sqrt{kt} + \beta^2 k\right) 
      = \sqrt{kt} \left(\sqrt{\frac{t}{k}} - 2\beta\right) - \beta^2 k 
      > \sqrt{k^2 (1 + \epsilon)} \left(\sqrt{1 + \epsilon} - 2\beta\right) - \beta^2 k \\
      &= (1 + \epsilon) k - 2\beta k\sqrt{1 + \epsilon} - \beta^2 k 
      = (1 + \epsilon)k - \left(2\beta\sqrt{1 + \epsilon} - \beta^2\right)k 
      = (1 + \epsilon)k - \frac{\epsilon}{2}k \\
      &\ge (1 + \epsilon/2)k
   \end{align*}

   So, using Theorem \ref{thm:qae_counting}
   with parameters $(M, 6)$ gives correct result with probability of at least
   $1 - 1/(2(6-1)) = 9/10$.

   Finally, we analyze the time complexity of this algorithm.
   From Theorem \ref{thm:qae_counting}, we know that 
   our algorithm queries $X$ and $Y$ at most $\tilde{O}(M)$ times.
   \begin{align*}
      M &= \left\lceil \frac{6\pi \sqrt{N / k}}{\sqrt{1 + 3\epsilon/2} - \sqrt{1 + \epsilon}} \right\rceil 
        \le 1 +  \frac{6\pi \sqrt{N / k}}{\sqrt{1 + 3\epsilon/2} - \sqrt{1 + \epsilon}} 
        \le 1 + \frac{6\pi \sqrt{2m / k}}{\sqrt{1 + 3\epsilon/2} - \sqrt{1 + \epsilon}} \\
        &\le 1 + \frac{6\pi \sqrt{2} \sqrt{m / k}}{\sqrt{1 + 3\epsilon/2} - \sqrt{1 + \epsilon}} 
        = O\left( \frac{\sqrt{\frac{m}{k}}}{\sqrt{1 + \frac{3}{2}\epsilon} - \sqrt{1 + \epsilon}} \right)
   \end{align*}

   A little algebra shows that $\frac{1}{\beta} \le 6\epsilon^{-1}$
   because $0 < \epsilon \le 1$:
   \begin{align*}
      \frac{1}{2}\epsilon
      &=
      \left( 1 + \frac{3}{2}\epsilon \right) - \left(1 + \epsilon\right) 
      =
      \left( \sqrt{1 + \frac{3}{2}\epsilon} + \sqrt{1 + \epsilon} \right)
      \left( \sqrt{1 + \frac{3}{2}\epsilon} - \sqrt{1 + \epsilon} \right) \\
      &\le
      \left(\sqrt{\frac{5}{2}} + \sqrt{2} \right)
      \left( \sqrt{1 + \frac{3}{2}\epsilon} - \sqrt{1 + \epsilon} \right) 
      \le
      3 \left( \sqrt{1 + \frac{3}{2}\epsilon} - \sqrt{1 + \epsilon} \right)  \\
      \frac{\epsilon}{6}
      &\le
      \sqrt{1 + \frac{3}{2}\epsilon} - \sqrt{1 + \epsilon} \\
      6\epsilon^{-1}
      &\ge
      \frac{1}{\sqrt{1 + \frac{3}{2}\epsilon} - \sqrt{1 + \epsilon}}
   \end{align*}

   Thus, the complexity 
   of the overall algorithm becomes $\tilde{O}\left( \epsilon^{-1} \sqrt{\frac{m}{k}} \right)$.
\end{proof}

Finally, we reach our main result.

\begin{theorem}
   \label{thm:approximate_epsilon}
   There exists a quantum algorithm that, given oracle access to 
   a pattern $P$ of length $m$ and
   a text $T$ of length $n$, an integer threshold $k > 0$, and
   $\epsilon \in (0, 1]$,
   such that:
   \begin{itemize}
      \item if there exists an $j \in [0..n-m]$ such that $\delta_H(T[j..j+m), P) \le k$, then 
         the algorithm, upon measurement, 
         outputs $(j', 1)$ for some $j' \in [0..n-m]$ satisfying 
         $\delta_H(T[j'..j'+m), P) \le (1 + \epsilon)k$ with a probability of at least $2/3$;
      \item if, for all $j \in [0..n-m]$, we have $\delta_H(T[j..j+m), P) > (1 + \epsilon)k$, then
         the algorithm, upon measurement, outputs $(j', 0)$ for some $j' \in [0..2n-1]$ 
         with probability of at least $2/3$.
   \end{itemize}
   This algorithm has time complexity $\tilde{O} \left( \epsilon^{-1} \sqrt{\frac{mn}{k}} \right)$ 
   (assuming that $P$ and $T$ can be accessed in $\tilde{O}\left(1\right)$ time).
\end{theorem}

\begin{proof}
   First, we present the quantum algorithm (Algorithm \ref{alg:approx_bounded_matching}):
   \begin{algorithm}
      \caption{$\mathbf{ApproxBoundedDistMatching} (T, P, k, \epsilon)$}\label{alg:approx_bounded_matching}
      \begin{algorithmic}[1]
         \STATE Set $n \leftarrow |T|$, $m \leftarrow |P|$.
         \STATE Set $N \leftarrow \min\{ 2^j : j \in \mathbb{N} \land 2^j \ge n - m + 1 \}$.
         \STATE \textbf{procedure} \textsc{Decider}(j)
         \STATE \hspace{\algorithmicindent} \textbf{if} $j > n - m$ \textbf{then}
         \STATE  \hspace{2\algorithmicindent} \textbf{return} $0$.
         \STATE \hspace{\algorithmicindent} \textbf{else}
         \STATE  \hspace{2\algorithmicindent}  \textbf{return} $\mathbf{ApproxBoundedHammingDecider}(T[j..j+m), P, k, \epsilon)$.
         \STATE \hspace{\algorithmicindent}\textbf{end if}
        \STATE \textbf{end procedure}
        \RETURN $\mathbf{Weak\_Search}(\textsc{Decider}, N)$.
      \end{algorithmic}
   \end{algorithm}

   Define $F : [0..N-1] \rightarrow \{ 0, 1, 2 \}$
   by letting, for $j \in [0..N-1]$,
   \begin{align*}
      F(j) = \begin{cases}
         0 & j > n - m \lor \delta_H(T[j..j+m-1], P) > (1 + \epsilon)k \\
         1 & \delta_H(T[j..j+m-1], P) \le k  \\
         2 & \text{otherwise}
      \end{cases}
   \end{align*}.

   From Theorem \ref{thm:epsilon_k_distance}, it is clear that
   for $j \in [0..N-1]$, $F(j) = 1$ implies that $\textsc{Decider}$
   returns $1$ with probability of at least $2/3$ and $F(j) = 0$ implies that
   $\mathcal{D}$ returns $0$ with probability of at least $2/3$.

   Thus, applying Algorithm \ref{alg:qtsearch}
   with $\textsc{Decider}$ and $F$,
   we get our desired quantum algorithm with 
   time complexity $\tilde{O}\left(\epsilon^{-1} \sqrt{\frac{mn}{k}}\right)$.
\end{proof}

\section{Further Direction}

What we have done is, simply speaking, just optimized bruteforce. 
There are methods shown in \cite{jin2022quantum} and \cite{kociumaka2024quantum}
to reduce the search space with $\tilde{O}\left( \sqrt{kn} \right)$-time preprocessing.
When $k = \Theta(m)$, using this slows down our algorithm.
As we are dealing with an additional approximation factor $\epsilon$,
can it be possible to bring the pre-processing time down?

\begin{acks}
    To Hasib sir and our mothers.
\end{acks}

\bibliographystyle{ACM-Reference-Format}
\bibliography{references}

%

\end{document}